\documentclass[11pt,a4paper]{article}

\usepackage[T1]{fontenc}
\usepackage[utf8]{inputenc}
\usepackage[margin=2.6cm]{geometry}
\usepackage{amsmath,amssymb,amsthm}
\usepackage{hyperref}
\hypersetup{colorlinks=true,linkcolor=blue,citecolor=blue,urlcolor=blue}

\newtheorem{theorem}{Theorem}
\newtheorem{lemma}[theorem]{Lemma}
\newtheorem{proposition}[theorem]{Proposition}

\theoremstyle{definition}
\newtheorem{definition}[theorem]{Definition}

\newcommand{\C}{\mathrm{C}}
\newcommand{\N}{\mathbb{N}}
\renewcommand{\Pr}{\mathbf{Pr}}
\newcommand{\eps}{\varepsilon}
\renewcommand{\le}{\leqslant}
\renewcommand{\ge}{\geqslant}

\title{Maximal Kolmogorov Complexity in a Hamming Ball}
\author{Alexander Kozachinskiy\thanks{Centro Nacional de Inteligencia
Artificial (CENIA).}
\and Nikolay Vereshchagin\thanks{Moscow State University, HSE University,
Yandex. The research was carried out within the framework of the scientific
program at the Faculty of Mechanics and Mathematics, Lomonosov Moscow State
University.}}
\date{}

\begin{document}
\maketitle

\begin{abstract}
The \emph{minimal} Kolmogorov complexity of a string within Hamming distance $r$
of a given string $x$ is the algorithmic rate--distortion function of $x$, and
Vereshchagin and Vit\'anyi characterized completely which shapes it can have.
This paper is about the opposite extreme. For a binary string $x$ of length $n$
let $g_x(r)$ denote the \emph{maximal} Kolmogorov complexity of a string within
Hamming distance $r$ of $x$; we study which values, and more generally which
functions of $r$, this quantity can attain.

First we characterize, up to an additive error $O(\log n)$, the possible values
of the triple $(\C(x),r,g_x(r))$: writing $r_k$ for the radius of a Hamming ball
of cardinality about $2^{k}$, a triple $(k,r,l)$ is realizable if and only if
$\log V(r_k+r)\le l\le\min\{n,\,k+\log V(r)\}$, where $V(a)$ is the
cardinality of a ball of radius $a$. In particular, for $r_k+r>n/2$ both bounds
collapse to $n$ and only $l\approx n$ is realizable. The two ends of this
interval correspond to the two extreme ways of placing a set of complexity $k$
in the cube: a single Hamming ball, where the lower bound comes from Harper's
isoperimetric inequality, and an error-correcting code, which for the
intermediate parameters we relax to a family of centers with bounded covering
multiplicity, in the spirit of list decoding. Then we turn to the function
$r\mapsto g_x(r)$ as a whole: we establish four properties that it always has,
and show that the minimal and the maximal functions consistent with these
properties are both attained, for every complexity level $k$. Which intermediate
profiles are attainable remains open.
\end{abstract}

\section{Introduction}

Let $x$ be a binary string of length $n$. All strings in this paper are binary
of length $n$, and $\C$ denotes the plain Kolmogorov complexity (see
Section~\ref{sec:prelim}). Write $B_r(x)$ for the Hamming ball of radius $r$
centered at $x$, and consider the complexities of the strings it contains.

The \emph{minimal} complexity in $B_r(x)$ is a well-studied quantity. As a
function of $r$ it is the algorithmic rate--distortion function of the
individual string $x$ for the Hamming distortion, the individual-string
counterpart of Shannon's rate--distortion function~\cite{Shannon59,CT06}. It is
a variant of Kolmogorov's structure function and belongs to algorithmic
statistics~\cite{GTV01,VV04,VS17}; Vereshchagin and Vit\'anyi~\cite{VV10}
characterized completely which shapes this curve can have, and read it as a
denoising procedure for $x$: the string of minimal complexity in the ball is the
natural candidate for the ``denoised'' version of $x$ at distortion level $r$.

This paper is about the opposite extreme. Define
\[
  g_x(r)\ =\ \max\{\C(y): y\in B_r(x)\}.
\]
Thus $g_x(0)=\C(x)$, and $g_x$ is a non-decreasing function reaching $n-O(1)$ at
$r=n$. We call $g_x$ the \emph{complexity profile} of $x$ and ask: what can this
profile look like? In contrast with the minimum, this question seems not to have
been studied systematically.

It is related to the effect of random noise on complexity, studied by Posobin
and Shen~\cite{PS19}: they show that a \emph{typical} corruption of a simple
string increases its complexity. We ask about the extreme rather than the
typical behaviour: how complex can the most complex string in the ball be? The
answer turns out to depend on $x$ and not only on $\C(x)$, which is what makes
the profile an interesting object.

We prove two kinds of results: a pointwise one, describing the possible values
of $g_x(r)$ for a single radius $r$, and global ones, describing possible
profiles as functions of $r$.

\paragraph{Pointwise results.}
Call a triple $(k,r,l)$ of natural numbers \emph{realizable} if there is a
string $x$ of length $n$ with $\C(x)=k$ and $g_x(r)=l$. Let
$V(a)=\sum_{i\le a}\binom ni$ be the cardinality of a Hamming ball of radius
$a$, and for $k\le n$ let $r_k$ denote the least radius with
$V(r_k)\ge2^{k-1}$, so that $\log V(r_k)=k+O(\log n)$ and, conveniently,
$r_k\le\lfloor n/2\rfloor$ (Section~\ref{sec:prelim}). Our first main
result characterizes the realizable triples up to $O(\log n)$-changes of the
parameters.

\begin{theorem}\label{th:triples}
There is a constant $c$ (independent of $n$) such that for all $n$ and all
triples $(k,r,l)$ of natural numbers with $k\le n$ the following hold.
\begin{enumerate}
\item[(a)] \textup{(Necessity)}\footnote{In both parts, one may write $k$
  instead of $\log V(r_k)$: by~\eqref{eq:rk} the two differ by at most
  $\log(n+1)$.} If $(k,r,l)$ is realizable, then
  \[
    \log V(r_k+r)-c\log n\ \le\ l\ \le\
    \min\{n,\,\log V(r_k)+\log V(r)\}+c\log n .
  \]
\item[(b)] \textup{(Sufficiency)} If
  \[
    \log V(r_k+r)\ \le\ l\ \le\ \min\{n,\,\log V(r_k)+\log V(r)\},
  \]
  then some triple $(k'',r,l'')$ with $|k''-k|\le c\log n$ and $|l''-l|\le
  c\log n$ is realizable.
\end{enumerate}
\end{theorem}

In the loosest form, ignoring all additive $O(\log n)$ terms (in the
inequalities and in the parameters), Theorem~\ref{th:triples} says that a
triple $(k,r,l)$ is realizable if and only if
\[
  \log V(r_k+r)\ \le\ l\ \le\ \min\{n,\;\log V(r_k)+\log V(r)\} .
\]
The interval here is never empty.
Indeed, $V(a)\,V(b)\ge V(a+b)$ for all $a,b$ (a string of weight at most $a+b$
splits into a part of weight at most $a$ and a part of weight at most $b$, and
this map is injective), so the right end is at least the left end. 
For $r_k+r\ge n/2$ both ends collapse to $n$: indeed, $\log V(r_k+r)\ge n-1$,
since a ball of radius $\lfloor n/2\rfloor$ contains at least half of the cube.

Informally: for a string of complexity $k$, the maximal complexity at distance
$r$ can be anything between $\log V(r_k+r)$ (as for a random string of a ball of
radius $r_k$ centered at $0^n$) and $k+\log V(r)$ (as for a random codeword of a
code correcting $r$ errors~\cite{MS77}), and these are the only constraints. The
upper bound is straightforward and the lower bound follows from Harper's
isoperimetric inequality. The sufficiency is proved by a construction
interpolating between the two extreme examples above: a union of $M$ balls of
radius $q$ centered at codewords, with $q$ ranging from $0$ to $r_k$. For the
intermediate parameters a code with disjoint balls need not exist, and we
replace it by a list-decoding relaxation~\cite{Elias57,Guruswami04}: a family of
centers whose balls cover every point $O(n)$ times (Lemma~\ref{lem:pack}). 

\paragraph{Global results.}
The profile $g_x$ always has the following four properties
(Section~\ref{sec:profile}): it is non-decreasing; $g_x(0)=\C(x)$ and
$g_x(r)\le n+O(1)$; it grows at most as fast as
\[
  g_x(r+a)\ \le\ g_x(r)+\log V(a)+O(\log n);
\]
and at least as fast as
\[
  g_x(r)\ge\log V(m)+2\log n+O(1)\ \Longrightarrow\ g_x(r+a)\ \ge\ \log V(m+a).
\]
For a given complexity level $k=g_x(0)$, the minimal function consistent with
these properties is $g_{\min}(r)=\log V(r_k+r)$, and the maximal one is
$g_{\max}(r)=\min\{n,\,k+\log V(r)\}$. We show
(Section~\ref{sec:extremal}) that both are attained for every $k$: there are
strings of complexity $k+O(\log n)$ whose profile is $g_{\min}+O(\log n)$, and
strings whose profile is $g_{\max}+O(\log n)$ --- \emph{for all radii
simultaneously}. Theorem~\ref{th:triples} controls the profile at a single
radius; whether every function satisfying the four properties is a profile of
some string remains open (Section~\ref{sec:open}).

\section{Preliminaries}\label{sec:prelim}

$\C(y)$ denotes the plain Kolmogorov complexity of $y$ with respect to a fixed
optimal description method; $\C(y\mid z)$ is the conditional version. We use the
symmetry of information: $\C(x,y)=\C(x)+\C(y\mid x)+O(\log n)=\C(y)+\C(x\mid
y)+O(\log n)$ for strings of length $n$. See~\cite{LV08,SUV} for background. The
length $n$ is treated as known: all our descriptions include $n$ and a constant
number of parameters not exceeding $n$, which is accounted for by the
$O(\log n)$ terms. $\log$ denotes the binary logarithm.

The Hamming ball of radius $r$ centered at $x$ is denoted $B_r(x)$, and
$B_r(X)=\bigcup_{x\in X}B_r(x)$ for a set $X$. We write
$V(a)=\sum_{i\le a}\binom ni$ for the cardinality of a ball of radius $a$ (it
depends on $n$, which is always clear from the context); the definition makes
sense for every $a\ge0$, with $V(a)=2^n$ for $a\ge n$. For $0\le k\le n$ we let
\[
  r_k\ =\ \min\{r:\ V(r)\ge2^{k-1}\}
\]
be the radius at which a ball first reaches the volume $2^{k-1}$,
so that 
$r_0=r_1=0$. Two properties of $r_k$ are used
throughout. First,
\begin{equation}\label{eq:rk}
  k-1\ \le\ \log V(r_k)\ \le\ k-1+\log(n+1),
\end{equation}
by the definition of $r_k$ and because
$V(r_k)\le(n+1)V(r_k-1)<(n+1)2^{k-1}$ by~\eqref{eq:ratio} if $r_k\ge1$, while
for $r_k=0$ we have $\log V(r_k)=0\le k-1+\log(n+1)$, as $n\ge1$. Second,
\begin{equation}\label{eq:rkhalf}
  r_k\ \le\ r_n\ =\ \lfloor n/2\rfloor, 
\end{equation}
since $r_k$ is non-decreasing in $k$, while
$V(\lfloor n/2\rfloor)\ge2^{n-1}>V(\lfloor n/2\rfloor-1)$: for odd $n$ these
are the identities $V(\frac{n-1}2)=2^{n-1}$ and
$V(\frac{n-1}2-1)=2^{n-1}-\binom n{(n-1)/2}$, and for even $n$ they read
$V(\frac n2)=\frac12\bigl(2^n+\binom n{n/2}\bigr)$ and
$V(\frac n2-1)=\frac12\bigl(2^n-\binom n{n/2}\bigr)$.

We will repeatedly use two standard facts. The first is (a corollary of)
Harper's vertex-isoperimetric theorem~\cite{Harper66}: \emph{among the subsets
of $\{0,1\}^n$ of a given cardinality, Hamming balls have the smallest
neighborhoods}. See~\cite{FF81} for a short proof and~\cite{Harper04} for a
systematic treatment of combinatorial isoperimetric problems. We need it in the
following form:
\begin{equation}\label{eq:harper}
  |A|\ \ge\ V(m)\quad\Longrightarrow\quad |B_a(A)|\ \ge\ V(m+a).
\end{equation}
The second concerns the growth of $V$ and  log-concavity of the sequence $V(0),V(1),\dots,V(n)$:
\begin{lemma}\label{l:growth}
\begin{equation}\label{eq:ratio}
  \frac{V(s)}{V(s-1)}\le n+1\, \text{ for all }s\ge1,
  \qquad
  \frac{V(s)}{V(s-1)}\ge1+\frac1n\ \text{ for }1\le s\le n/2 .
\end{equation}
\begin{equation}\label{eq:logconc}
  V(s)^2\ \ge\ V(s-1)\,V(s+1)\qquad(1\le s\le n-1),
\end{equation}
that is, $V(s+1)/V(s)$ is non-increasing in $s$; consequently, the ratio
$V(s+a)/V(s)=\prod_{j<a}V(s+j+1)/V(s+j)$ is non-increasing in $s$ for every
fixed $a\ge0$. 
\end{lemma}
The proofs of these technical inequalities are deferred to Section~\ref{proof:l2}.

\section{Properties of the profile}\label{sec:profile}

\begin{proposition}\label{prop:props}
For every $x$ of length $n$:
\begin{enumerate}
\item[(P1)] $g_x$ is non-decreasing;
\item[(P2)] $g_x(0)=\C(x)$ and $g_x(r)\le n+O(1)$ for all $r$;
\item[(P3)] $g_x(r+a)\le g_x(r)+\log V(a)+O(\log n)$;
\item[(P4)] if $g_x(r)\ge\log V(m)+2\log n+O(1)$, then
  $g_x(r+a)\ge\log V(m+a)$.
\end{enumerate}
\end{proposition}

\begin{proof}
(P1) and (P2) are immediate. For (P3), let $y$ be any string at distance at most
$r+a$ from $x$; we bound $\C(y)$. There is a string $z$ at distance at most $a$
from $y$ and at most $r$ from $x$ (move from $y$ towards $x$ by $a$ steps along
positions where they differ). Then $z$ can be described by $g_x(r)$ bits, and
$y$ can be recovered from $z$ and the index of $y$ in $B_a(z)$, which takes
$\log V(a)$ more bits, plus $O(\log n)$ bits of overhead.

For (P4), let $S=\{z:\C(z)<\log V(m+a)\}$, so $|S|<V(m+a)$, and let
\[
  Y\ =\ \{y: B_a(y)\subseteq S\} .
\]
First, $|Y|< V(m)$: otherwise $|Y|\ge V(m)$ and, by
Harper~\eqref{eq:harper}, $|B_a(Y)|\ge V(m+a)>|S|$, contradicting
$B_a(Y)\subseteq S$. Second, $Y$ can be enumerated given $n$, $a$ and the
threshold $\log V(m+a)$: the set $S$ is enumerable, and the condition
$B_a(y)\subseteq S$ is monotone in $S$. Hence every element of $Y$ is determined
by its index in this enumeration and the parameters, so
\begin{equation}\label{eq:y}
  y\in Y\ \Longrightarrow\ \C(y)\ \le\ \log V(m)+2\log n+O(1).
\end{equation}

For the sake of contradiction assume $g_x(r+a)<\log V(m+a)$. Then for every $y\in B_r(x)$ we have
$B_a(y)\subseteq B_{r+a}(x)\subseteq S$, that is, $B_r(x)\subseteq Y$; hence by~\eqref{eq:y} for all $y\in B_r(x)$ we have $\C(y)\ \le\ \log V(m)+2\log n+O(1)$. In other words,
$g_x(r)\le\log V(m)+2\log n+O(1)$, contradicting the assumption.
\end{proof}

\medskip
\emph{The proof of part (a) of
Theorem~\ref{th:triples}}.
We may assume that $k\ge2\log n+O(1)$, with the constant from (P4). Otherwise
$\log V(r_k)=O(\log n)$ by~\eqref{eq:rk}, and the left inequality of
Theorem~\ref{th:triples}(a) is trivial: the ball $B_a(x)$ has $V(a)$ elements,
so $g_x(a)\ge\log V(a)$, while
$\log V(r_k+a)\le\log V(r_k)+\log V(a)=\log V(a)+O(\log n)$ by the
splitting inequality $V(b)V(c)\ge V(b+c)$.

So let $x$ be a string with $\C(x)=k$ and let $a$ be a radius; we have to prove
the two inequalities of Theorem~\ref{th:triples}(a) for $l=g_x(a)$.

\emph{The upper bound.} Apply (P3) with $r=0$: as $g_x(0)=\C(x)=k$ by (P2),
\[
  g_x(a)\ \le\ k+\log V(a)+O(\log n) .
\]
Together with the bound $g_x(a)\le n+O(1)$, also from (P2), this gives
$g_x(a)\le\min\{n,\,k+\log V(a)\}+O(\log n)$. Finally, $k\le\log V(r_k)+1$ by
the left inequality of~\eqref{eq:rk}, so $k$ may be replaced here by
$\log V(r_k)$, and we obtain the upper bound of the theorem.

\emph{The lower bound.} Here we use (P4), again with $r=0$: it says that if
$k=g_x(0)\ge\log V(m)+2\log n+O(1)$, then $g_x(a)\ge\log V(m+a)$. Naturally we
apply it to the largest admissible $m$, that is, to
\[
  m\ =\ \max\{m':\ \log V(m')+2\log n+O(1)\le k\}
\]
(the set is non-empty: $m'=0$ is admissible, since $\log V(0)=0$ and
$k\ge2\log n+O(1)$ by our assumption). This gives
\begin{equation}\label{eq:pointwise}
  g_x(a)\ \ge\ \log V(m+a) ,
\end{equation}
and it remains to show that $\log V(m+a)\ge\log V(r_k+a)-O(\log n)$.

The radii $m$ and $r_k$ are close to each other, because $m$ is defined by the
same recipe as $r_k$, only for a complexity level lowered by $2\log n+O(1)$.
Precisely, let $k'=k-2\log n-O(1)+1$ with the same constant, so that
$k-k'=O(\log n)$ and $k'\le k$. By the maximality of $m$ we have
$\log V(m+1)>k-2\log n-O(1)=k'-1$, that is, $V(m+1)\ge2^{k'-1}$, and therefore
$r_{k'}\le m+1$ by the definition of $r_{k'}$. Now we compare $V(m+a)$ with
$V(r_k+a)$ in two steps.

\emph{From $m$ to $r_{k'}$.} As $m\ge r_{k'}-1$ and $V$ is non-decreasing,
\[
  \log V(m+a)\ \ge\ \log V(r_{k'}+a-1)\ \ge\ \log V(r_{k'}+a)-\log(n+1)
\]
by the first inequality of~\eqref{eq:ratio} (and trivially so if
$r_{k'}+a=0$).

\emph{From $r_{k'}$ to $r_k$.} Both radii are shifted by the same $a$, and the
point is that the shift can only decrease the multiplicative gap between the
two volumes. Indeed, let $d=r_k-r_{k'}\ge0$ (the radius $r_k$ is non-decreasing
in $k$, and $k'\le k$). By log-concavity~\eqref{eq:logconc} the ratio
$V(s+d)/V(s)$ is non-increasing in $s$; comparing its values at $s=r_{k'}+a$
and at $s=r_{k'}$, we get
\[
  \frac{V(r_k+a)}{V(r_{k'}+a)}\ =\ \frac{V(r_{k'}+a+d)}{V(r_{k'}+a)}\ \le\
  \frac{V(r_{k'}+d)}{V(r_{k'})}\ =\ \frac{V(r_k)}{V(r_{k'})}\ \le\
  \frac{(n+1)2^{k-1}}{2^{k'-1}}\ =\ (n+1)\,2^{k-k'} ,
\]
the last step by the right inequality of~\eqref{eq:rk} applied to $r_k$ and the
left one applied to $r_{k'}$. Hence
\[
  \log V(r_{k'}+a)\ \ge\ \log V(r_k+a)-(k-k')-\log(n+1)\ =\
  \log V(r_k+a)-O(\log n) .
\]

Combining the two steps with~\eqref{eq:pointwise} we get
$g_x(a)\ge\log V(r_k+a)-O(\log n)$, which is the lower bound of the theorem.
This proves part (a) of Theorem~\ref{th:triples}.

\section{Realizable triples}\label{sec:triples}

In this section we prove part (b) of Theorem~\ref{th:triples}. The construction
is based on the following notion.

\begin{definition}
Let $k,r,l\in\N$ and $\eps\in(0,1)$. A set $C\subseteq\{0,1\}^n$ is a
\emph{$(k,r,l,\eps)$-concentrator} if
\begin{enumerate}
\item[(C1)] $\eps\,2^{k}\le|C|\le 2^{k}$;
\item[(C2)] $\eps\,2^{l}\le|B_r(C)|\le 2^{l}$;
\item[(C3)] for every $C'\subseteq C$ with $|C'|\ge(1-\eps)|C|$ we have
  $|B_r(C')|\ge\eps\,|B_r(C)|$.
\end{enumerate}
\end{definition}

Thus a concentrator is a set of about $2^k$ strings whose $r$-neighborhood has
about $2^l$ elements and, moreover, cannot be shrunk much by discarding a small
fraction of $C$. The parameter $\eps$ will be $1/\mathrm{poly}(n)$, and all
occurrences of $\eps$ below are absorbed by $O(\log(1/\eps))=O(\log n)$ terms.

\begin{proposition}\label{prop:conc}
If a $(k,r,l,\eps)$-concentrator exists, then some triple $(k'',r,l'')$ with
$|k''-k|=O(\log(n/\eps))$ and $|l''-l|=O(\log(n/\eps))$ is realizable.
\end{proposition}

\begin{proof}
Given the parameters $n,k,r,l$ and $\eps$ (we may round $\eps$ down to a power
of two, so the parameters take $O(\log(n/\eps))$ bits), we can find some
$(k,r,l,\eps)$-concentrator $C$ by exhaustive search: all conditions are finite.
Then:
\begin{itemize}
\item every $x\in C$ satisfies $\C(x)\le k+O(\log(n/\eps))$ ($x$ is determined
  by the parameters and its index in $C$), and likewise every $y\in B_r(C)$
  satisfies $\C(y)\le l+O(\log(n/\eps))$;
\item let $C'$ be the set of elements of $C$ of complexity at least
  $k-2\log(1/\eps)$. Fewer than $\eps^2 2^k$ strings altogether have smaller
  complexity, and $|C|\ge\eps2^k$ by (C1), so
  $$|C\setminus C'|\le\eps^22^k\le\eps|C|$$ and hence
  $|C'|\ge(1-\eps)|C|$. By (C3) and (C2),
  $|B_r(C')|\ge\eps|B_r(C)|\ge\eps^22^{l}$, while fewer than $\eps^22^l$
  strings have complexity below $l-2\log(1/\eps)$; so some $y_0\in B_r(C')$
  has $\C(y_0)\ge l-2\log(1/\eps)$.
\end{itemize}
The string $y_0$ lies in $B_r(x)$ for some $x\in C'$. For this $x$ we have
$\C(x)\in[k-2\log(1/\eps),\,k+O(\log(n/\eps))]$ and
$g_x(r)\in[l-2\log(1/\eps),\,l+O(\log(n/\eps))]$ (the upper bound because
$B_r(x)\subseteq B_r(C)$). So the triple $(\C(x),r,g_x(r))$ is as required.
\end{proof}

It remains to construct concentrators for the whole range of parameters. The
extreme cases are a single Hamming ball ($l=\log V(r_k+r)$) and an
error-correcting code ($l=k+\log V(r)$); the general construction interpolates
between them. Instead of codes with disjoint balls, whose existence is
problematic for some parameters, we use the following relaxation. It is the
counting counterpart of list decoding~\cite{Elias57,Guruswami04}: bounding the
covering multiplicity by $8n$ is exactly saying that every received word has at
most $8n$ candidate codewords at distance $s$. 

\begin{lemma}[centers with small covering multiplicity]\label{lem:pack}
Let $s\ge0$ be any radius and let $M$ be such that
$M\cdot V(s)\le 2n\cdot2^n$. Then there exist $x_1,\dots,x_M\in\{0,1\}^n$ such that every $y\in\{0,1\}^n$ belongs to at
most $8n$ of the balls $B_s(x_i)$.
\end{lemma}

\begin{proof}
Choose $x_1,\dots,x_M$ independently and uniformly at random, and fix a string
$y$. For each $i$ the event ``$y\in B_s(x_i)$'' has probability
\[
  p\ :=\ \frac{|B_s(y)|}{2^n}\ =\ \frac{V(s)}{2^n}
\]
(indeed, $y\in B_s(x_i)$ if and only if $x_i\in B_s(y)$), and these $M$ events
are independent. By the assumption of the lemma, $Mp\le2n$.

If $y$ is covered more than $8n$ times, then all the events with indices in
some set $I\subseteq\{1,\dots,M\}$ of size $8n$ occur; for a fixed $I$ this has
probability $p^{8n}$, so the union bound over all such $I$ gives
\begin{multline*}
  \Pr[\text{$y$ is covered more than $8n$ times}]\ \le\
  \binom{M}{8n}p^{8n}\ \le\ \Bigl(\frac{eM}{8n}\Bigr)^{8n}p^{8n}\\
  =\ \Bigl(\frac{e\,Mp}{8n}\Bigr)^{8n}\ \le\
  \Bigl(\frac{e}{4}\Bigr)^{8n}\ <\ 2^{-n},
\end{multline*}
the last but one inequality because $Mp\le2n$, and the last one because
$(e/4)^8<1/2$. The second inequality is the standard bound $\binom Mt\le(eM/t)^t$.
It comes from the estimate
$t!\ge(t/e)^t$, which in turn is a single term of the exponential series:
$e^t=\sum_{j\ge0}t^j/j!\ge t^t/t!$; together with $\binom Mt\le M^t/t!$ this
gives $\binom Mt\le M^t(e/t)^t$.

Finally, by the union bound over all $2^n$ strings $y$, with positive
probability no string is covered more than $8n$ times.
\end{proof}

\begin{lemma}\label{lem:construct}
Let $q\le n/2$ be such that $V(q)\le2^k$. Set  $M:=\lfloor2^k/V(q)\rfloor\ge1$.
Assume further that 
$M\cdot V(q+r)\le2n\cdot2^n$. Then there is a
$\bigl(k,\,r,\,\lceil\log (M V(q+r))\rceil,\,1/32n^2\bigr)$-concentrator.
\end{lemma}

\begin{proof}
Let $x_1,\dots,x_M$ be as in Lemma~\ref{lem:pack} with $s=q+r$: every string is
covered by at most $8n$ of the balls $B_{q+r}(x_i)$; in particular it is covered
by at most $8n$ of the smaller balls $B_{q}(x_i)$. Put
\[
  C\ =\ \bigcup_{i=1}^{M}B_q(x_i),\qquad\text{so}\qquad
  B_r(C)\ =\ \bigcup_{i=1}^{M}B_{q+r}(x_i).
\]

\emph{Sizes.} Counting with multiplicities, $$MV(q)\ge|C|\ge MV(q)/(8n)$$ and
$$MV(q+r)\ge|B_r(C)|\ge MV(q+r)/(8n).$$ 
Here $MV(q)$ is close to $2^k$  by the choice of $M$. Indeed, $MV(q)\le2^k$ and
$$MV(q)\ge2^k-V(q)\ge2^{k-1}$$ (if $V(q)\le2^{k-1}$ this is clear, and otherwise
$M=1$ and $MV(q)=V(q)>2^{k-1}$). With $l:=\lceil\log(MV(q+r))\rceil$,
conditions (C1) and (C2) hold with $\eps\le1/16n$. Indeed
$$
|C|\le MV(q)\le 2^k \text{ and } 
|C|\ge MV(q)/8n\ge2^{k-1}/8n.
$$ 
Besides,  $2^{l-1}<MV(q+r)\le 2^l$ by the choice of $l$. Hence 
$$
|B_r(C)|\le MV(q+r)\le 2^l \text{ and } |B_r(C)|\ge MV(q+r)/8n>2^{l-1}/8n.
$$

\emph{Expansion.} Let $\eps=1/32n^2$ and let $C'\subseteq C$ satisfy
$|C'|\ge(1-\eps)|C|$; we have to bound $|B_r(C')|$ from below.

\emph{Step 1: most of the balls $B_q(x_i)$ lose few elements.} Consider the quantity
\[
  \Sigma\ :=\ \sum_{i=1}^{M}\bigl|B_q(x_i)\setminus C'\bigr| ,
\]
accounting for the total loss of all the balls, where 
each lost string $y$ is counted once for every index $i$ with $y\in B_q(x_i)$
--- so a string lying in several of the balls is counted several times.
 Every string counted here lies in
$C\setminus C'$, because $B_q(x_i)\subseteq C$, and it lies in at most $8n$ of
the balls $B_q(x_i)$. Hence
\[
  \Sigma\ =\ \sum_{y\in C\setminus C'}\#\{i:\ y\in B_q(x_i)\}\ \le\
  8n\,|C\setminus C'|\ \le\ 8n\eps|C|\ \le\ 8n\eps MV(q),
\]
the last step because $|C|\le MV(q)$. Now $\Sigma$ is a sum of $M$ non-negative
terms, so fewer than $M/2$ of them exceed $2\Sigma/M$; consequently for at
least $M/2$ indices $i$ (call them \emph{good})
\[
  |B_q(x_i)\setminus C'|\ \le\ \frac{2\Sigma}{M}\ \le\ 16n\eps V(q)\ =\
  \frac{V(q)}{2n},
\]
the last step because  $\eps=1/32n^2$.

\emph{Step 2: a good ball still contains more than $V(q-1)$ elements.} Since
$|B_q(x_i)|=V(q)$, for a good $i$ we get
\[
  |B_q(x_i)\cap C'|\ =\ V(q)-|B_q(x_i)\setminus C'|\ \ge\
  V(q)\Bigl(1-\frac1{2n}\Bigr)\ >\ V(q-1) .
\]
The last inequality is trivial for $q=0$ (then $V(q-1)=0$), and for $q\ge1$ it
holds because $V(q)\ge(1+1/n)V(q-1)$ by the second inequality
of~\eqref{eq:ratio} --- this is the only place where the hypothesis $q\le n/2$
is used --- so that $V(q-1)\le\frac n{n+1}V(q)$, while
$1-\frac1{2n}>\frac n{n+1}$ for $n\ge2$.

\emph{Step 3: Harper.} Applying~\eqref{eq:harper} to the set
$B_q(x_i)\cap C'$, whose cardinality is at least that of a ball of radius
$q-1$, we get for every good $i$
\[
  \bigl|B_r\bigl(B_q(x_i)\cap C'\bigr)\bigr|\ \ge\ V(q-1+r)\ \ge\
  \frac{V(q+r)}{n+1},
\]
the last step by the first inequality of~\eqref{eq:ratio}, which is valid for
all radii --- note that $q+r$ may well exceed $n/2$ here.

\emph{Step 4: from multiplicities back to the union.} For a good $i$ the set
$B_r(B_q(x_i)\cap C')$ is contained in $B_r(C')$, and also in $B_{q+r}(x_i)$.
Every string belongs to at most $8n$ of the balls $B_{q+r}(x_i)$ and hence to
at most $8n$ of these sets, so the union of the latter, which is a subset of
$B_r(C')$, is at least their total size divided by $8n$:
\[
  |B_r(C')|\ \ge\ \frac1{8n}\sum_{i\ \text{good}}
  \bigl|B_r\bigl(B_q(x_i)\cap C'\bigr)\bigr|\ \ge\
  \frac1{8n}\cdot\frac M2\cdot\frac{V(q+r)}{n+1}\ \ge\
  \frac{MV(q+r)}{32n^2} .
\]
Finally $|B_r(C)|\le MV(q+r)$ and $\eps=1/32n^2$, so
$|B_r(C')|\ge\eps\,|B_r(C)|$. This is (C3). 
\end{proof}

Before proving part~(b), let us say how Lemma~\ref{lem:construct} will be
used. It has one free parameter, the radius $q$ of the balls around the
centers, and it yields a concentrator whose third parameter is
\[
  \lambda(q)\ :=\ \log\bigl(\lfloor2^k/V(q)\rfloor\cdot V(q+r)\bigr) .
\]
For $q=0$ the construction is a code-like family of $2^k$ centers, and
$\lambda(0)=k+\log V(r)$ --- the upper end of the window in
Theorem~\ref{th:triples}(b). As $q$ grows, the number
$M=\lfloor2^k/V(q)\rfloor$ of centers drops, and for the largest admissible
$q$ the construction degenerates into a single ball of radius about $r_k$,
where $\lambda$ is close to $\log V(r_k+r)$ --- the lower end of the window.
Moreover, $\lambda$ changes slowly, by $O(\log n)$ per step. Hence, to realize
a given $l$ inside the window, it suffices to scan $q$ from $0$ upwards and
stop as soon as $\lambda(q)$ comes close to $l$; Proposition~\ref{prop:conc}
then turns the resulting concentrator into a realizable triple. This is what
the proof below does.

\begin{proof}[Proof of Theorem~\ref{th:triples}(b)]
Let 
$$
\log V(r_k+r)\le l\le\min\{n,\,\log V(r_k)+\log V(r)\}.
$$ 
Let $Q$ be the largest $q\le r_k$ with $V(q)\le2^k$; it exists, since $q=0$
qualifies, and $Q\le r_k\le\lfloor n/2\rfloor$ by~\eqref{eq:rkhalf}. For
$q\le Q$ we have $V(q)\le V(Q)\le2^k$, so the function
\[
  q\ \longmapsto\ \lambda(q)\ :=\
  \log\bigl(\lfloor2^k/V(q)\rfloor\cdot V(q+r)\bigr),
  \qquad q=0,1,\dots,Q ,
\]
is well defined, and its consecutive values differ by at most $2\log(n+1)$:
by~\eqref{eq:ratio} the factor $V(q+r)$ changes by at most $n+1$ in either
direction, and the factor $\lfloor2^k/V(q)\rfloor$ by at most $2(n+1)$, since
$\lfloor t\rfloor\ge t/2$ for $t\ge1$.

At the left end, $\lambda(0)=k+\log V(r)$, while
$l\le\log V(r_k)+\log V(r)\le k-1+\log(n+1)+\log V(r)$ by~\eqref{eq:rk}, so
\begin{equation}\label{eq:leftend}
  \lambda(0)\ \ge\ l-\log(n+1).
\end{equation}
At the right end,
\begin{equation}\label{eq:rightend}
  \lambda(Q)\ \le\ l+\log(n+1)\ \le\ n+\log(n+1).
\end{equation}
Indeed, $V(Q+r)\le V(r_k+r)$ because $Q\le r_k$. And $V(Q)\ge2^k/(n+1)$:
this is clear if $Q=r_k$, since $V(r_k)\ge2^{k-1}$, and if $Q<r_k$ then
$V(Q+1)>2^k$ by the maximality of $Q$, whence $V(Q)\ge V(Q+1)/(n+1)$
by~\eqref{eq:ratio}. Therefore
\[
  2^{\lambda(Q)}\ \le\ \frac{2^k}{V(Q)}\,V(Q+r)\ \le\ (n+1)\,V(r_k+r)
  \ \le\ (n+1)\,2^{l},
\]
the last step by the assumed lower bound $l\ge\log V(r_k+r)$. In particular
$2^{\lambda(Q)}\le(n+1)2^n\le2n\cdot2^n$.

So $\lambda$ starts above $l-\log(n+1)$ by~\eqref{eq:leftend} and ends below
$l+\log(n+1)$ by~\eqref{eq:rightend}, while its steps are at most
$2\log(n+1)$; hence $|\lambda(q^*)-l|\le\log(n+1)$ for some $q^*\le Q$, and
then $2^{\lambda(q^*)}\le(n+1)2^{l}\le2n\cdot2^n$.
Therefore Lemma~\ref{lem:construct} applies with $q=q^*$ and provides a
$(k,r,\lceil\lambda(q^*)\rceil,1/32n^2)$-concentrator;
Proposition~\ref{prop:conc} turns it into a realizable triple $(k'',r,l'')$
with $|k''-k|=O(\log n)$ and $|l''-l|=O(\log n)$.
\end{proof}


\section{The extremal profiles}\label{sec:extremal}

We now turn to the profile $g_x$ as a function of $r$. Fix a complexity level $k$ and
consider the class of functions satisfying (P1)--(P4) with $g(0)=k$. The
minimal such function is
\[
  g_{\min}(r)\ =\ \log V(r_k+r)
\]
(up to $O(\log n)$: by (P4) with $r=0$ the function must be at least this, and
$g_{\min}$ itself satisfies (P1)--(P4)), and the maximal one is
\[
  g_{\max}(r)\ =\ \min\{n,\;k+\log V(r)\}
\]
(by (P3) with $r=0$ and (P2)). Both are attained, for all radii simultaneously.

\begin{theorem}\label{th:gmin}
For every $k\le n$ there is a string $x$ of length $n$ with
$\C(x)=k+O(\log n)$ and
\[
  g_x(r)\ =\ \log V(r_k+r)+O(\log n)
  \qquad\text{for all $r$ simultaneously.}
\]
\end{theorem}

\begin{proof}
Let  $x$ be a string of maximal complexity in $B_{r_k}(0^n)$, so
that $\C(x)\ge\log V({r_k})-O(1)=k-O(\log n)$; also $\C(x)\le\log V({r_k})+O(\log
n)=k+O(\log n)$, since $x$ is determined by ${r_k}$ and its index in $B_{r_k}(0^n)$.

Upper bound on the profile: $B_r(x)\subseteq B_{{r_k}+r}(0^n)$, so every
$y\in B_r(x)$ is determined by ${r_k}+r$ and its index in $B_{{r_k}+r}(0^n)$, whence
$g_x(r)\le\log V({r_k}+r)+O(\log n)$.

Lower bound. As in the proof of Theorem~\ref{th:triples}(a) we may assume
$k\ge2\log n+O(1)$, the bound being trivial for smaller $k$. We have
$g_x(0)=\C(x)\ge\log V({r_k})-O(\log n)$, so by (P4), applied with the
maximal $m$ such that $\log V(m)+2\log n+O(1)\le\C(x)$, we get
$g_x(r)\ge\log V(m+r)$, where $m\ge r_{k'}-1$ for $k'=k-O(\log n)$; and
$\log V(m+r)\ge\log V(r_k+r)-O(\log n)$ by~\eqref{eq:ratio} and
log-concavity~\eqref{eq:logconc}, exactly as in the derivation of
Theorem~\ref{th:triples}(a).
\end{proof}

\begin{theorem}\label{th:gmax}
For every $k\le n$ there is a string $x$ of length $n$ with $\C(x)=k+O(\log n)$
and $g_x(r)=g_{\max}(r)+O(\log n)$ for all $r$ simultaneously.
\end{theorem}

We may assume $k\ge1$: for $k=0$ the string $x=0^n$ works, since
$g_{0^n}(r)=\log V(r)+O(\log n)=g_{\max}(r)+O(\log n)$. Let $r^*$ be the
maximal radius with $k+\log V(r^*)\le n$; then $r^*\le n-1$, as
$k+\log V(n)=k+n>n$. It suffices to
control the profile for $r\le r^*$: for larger $r$, monotonicity and (P2) give
$g_x(r)=n+O(\log n)$ automatically, provided $g_x(r^*)\ge n-O(\log n)$. The
latter does follow from the bound $g_x(r)\ge k+\log V(r)-O(\log n)$ proved
below for $r\le r^*$: by the maximality of $r^*$ we have $k+\log V(r^*+1)>n$,
while $V(r^*+1)\le(n+1)V(r^*)$ by~\eqref{eq:ratio}, so
$k+\log V(r^*)\ge n-\log(n+1)$.

We give two proofs: a short one via Lemma~\ref{lem:pack}, and an
enumeration-based one whose technique may be useful for the open problem of
Section~\ref{sec:open}.

\begin{proof}[First proof]
Apply Lemma~\ref{lem:pack} with $s=r^*$ and $M=\lfloor2^n/V(r^*)\rfloor$; note
$\log M=n-\log V(r^*)+O(1)=k+O(\log n)$. Let $x$ be a string of maximal
complexity among the centers $x_1,\dots,x_M$, so that
$\C(x)=\log M+O(\log n)=k+O(\log n)$ (the whole family can be found by
exhaustive search from the parameters, so each $x_i$ is determined by its index).

Fix $r\le r^*$ and pick $y\in B_r(x)$ with $\C(y\mid x)\ge\log V(r)$ (such $y$
exists by counting). The ball $B_{r^*}(y)$ contains $x$, and by the choice of
the family it contains at most $8n$ of the centers; therefore $\C(x\mid
y)\le2\log n+O(1)$ ($x$ is determined by $y$ and the index of $x$ among the
centers lying in $B_{r^*}(y)$). By symmetry of information,
\begin{multline*}
  \C(y)\ =\ \C(x,y)-\C(x\mid y)+O(\log n)\ =\
  \C(x)+\C(y\mid x)-\C(x\mid y)+O(\log n)\\
  \ge\ k+\log V(r)-O(\log n).
\end{multline*}
So $g_x(r)\ge k+\log V(r)-O(\log n)$; the matching upper bound is (P3) with
$r=0$.
\end{proof}

\begin{proof}[Second proof]
We construct $x$ by a series of attempts. Start with $x=0^n$. Enumerate all
pairs (string, its complexity bound) and maintain, for each threshold $i$, the
set $S_i$ of strings so far discovered to have complexity less than $i$.
Whenever it turns out that for some $r\le r^*$ the current $x$ satisfies
\[
  B_r(x)\ \subseteq\ S_{k+\log V(r)-\log2n}
  \tag{$*$}
\]
(all words in the ball are simple), we discard $x$ and choose a new one by the
following combinatorial lemma, applied to the current sets $S_i$.

\begin{lemma}\label{lem:half}
For any sets $S_i$ with $|S_i|<2^i$ there exists $x$ such that for every $r\le
r^*$ at most half of the ball $B_r(x)$ belongs to $S_{k+\log V(r)-\log2n}$.
\end{lemma}

\begin{proof}
Let $\eps_x(r)$ be the fraction of $B_r(x)$ lying in $S_{k+\log V(r)-\log2n}$.
Each fixed string belongs to $B_r(x)$ for a uniformly random $x$ with
probability $V(r)/2^n$, hence
\[
  \mathbf{E}\,|B_r(x)\cap S_{k+\log V(r)-\log2n}|
  \ =\ |S_{k+\log V(r)-\log2n}|\cdot\frac{V(r)}{2^n}
  \ \le\ \frac{2^kV(r)}{2n}\cdot\frac{V(r)}{2^n}
  \ \le\ \frac{V(r)}{2n},
\]
where the last step uses $2^kV(r)\le2^n$ (valid for $r\le r^*$). So
$\mathbf{E}\,\eps_x(r)\le1/2n$, and
$\mathbf{E}\sum_{r\le r^*}\eps_x(r)\le(r^*+1)/2n\le1/2$, the last step
because $r^*\le n-1$. Some $x$ realizes at most
the expectation, and for that $x$ the sum and hence every single term is at most $1/2$.
\end{proof}

Whenever ($*$) forces a change of $x$ prompted by a radius $r$, at least
$V(r)/2$ \emph{new} strings must have entered $S_{k+\log V(r)-\log2n}$ since
the previous change prompted by the same $r$ (at the moment of the previous
choice, at most half of the then-current ball was simple). Since
$|S_{k+\log V(r)-\log2n}|<2^kV(r)/2n$, the radius $r$ can prompt at most
$2^k/n$ changes, and in total there are at most $(r^*+1)\cdot2^k/n\le2^k$
changes.
The process is deterministic given $n,k,r^*$, and $x$ stabilizes; the final $x$
is determined by the parameters and the total number of changes, whence
$\C(x)\le k+O(\log n)$.

By construction, for the final $x$ and every $r\le r^*$ the ball $B_r(x)$
contains a string of complexity at least $k+\log V(r)-\log2n$; in particular
($r=0$) $\C(x)\ge k-\log2n$. So $\C(x)=k+O(\log n)$ and
$g_x(r)\ge k+\log V(r)-O(\log n)$ for all $r\le r^*$, while the upper bound is
(P3).
\end{proof}

\section{Open questions}\label{sec:open}

The natural remaining question is a characterization of the attainable
profiles: \emph{is every function satisfying (P1)--(P4) equal to $g_x+O(\log
n)$ (or at least $g_x+o(n)$) for some $x$?} Theorem~\ref{th:triples} pins the
profile at one radius, and its proof in fact controls the whole initial
segment: for the string $x$ produced from the concentrator of
Lemma~\ref{lem:construct} one can check that $g_x(r')=\log M+\log
V(q+r')+O(\log n)$ for \emph{all} $r'\le r$, which is a translate of the
minimal-growth curve. The difficulty is to glue different growth rates at
different scales; the enumeration technique of the second proof of
Theorem~\ref{th:gmax} seems to be a natural tool.

A second question is quantitative: our characterization has additive error
$O(\log n)$, and we did not try to optimize the constant in it. Is the error
$O(\log n)$ optimal for the pointwise problem?

\section{The proof of Lemma~\ref{l:growth}}\label{proof:l2}
Equation~\eqref{eq:ratio}. Indeed, $V(s)-V(s-1)=\binom ns\le n\binom n{s-1}\le nV(s-1)$ for every
$s\ge1$ (for $s>n$ both sides vanish), which gives the first inequality. For $s\le n/2$ the binomial coefficients
$\binom ni$ increase for $i\le s$; hence
$V(s-1)=\sum_{i<s}\binom ni\le s\binom ns\le n\binom ns=n(V(s)-V(s-1))$,
which gives the second one.

To check~\eqref{eq:logconc}, write $a_i=\binom ni$ and note
that
\begin{multline*}
  V(s)^2-V(s-1)V(s+1)\ =\
  V(s)\bigl(V(s)-V(s-1)\bigr)-V(s-1)\bigl(V(s+1)-V(s)\bigr)\\
  =\ a_sV(s)-a_{s+1}V(s-1).
\end{multline*}
Now
\[
  a_sV(s)\ =\ a_sa_0+\sum_{i=1}^{s}a_sa_i\ \ge\ \sum_{i=1}^{s}a_{s+1}a_{i-1}
  \ =\ a_{s+1}V(s-1),
\]
where the middle inequality is term-wise: $a_sa_i\ge a_{s+1}a_{i-1}$, because
the ratios $a_j/a_{j+1}$ are non-decreasing in $j$ (log-concavity of the
binomial coefficients) and $i-1<s$.

\end{document}